%% file: main.tex
\documentclass[hidelinks,onefignum,onetabnum]{siamonline250211_TOC}

\input{ex_shared}
\usepackage{booktabs}
\usepackage{nicefrac}

\DeclareMathOperator*{\minimize}{minimize}

\begin{document}

\maketitle  

\begin{abstract}
Wavelet phase is a critical parameter in seismic processing, where zero-phase wavelets are essential for maximizing temporal resolution and ensuring accurate interpretation of subsurface structures. In practice, however, the seismic wavelet is often nonstationary, exhibiting a phase that varies in space and time due to physical factors such as attenuation, dispersion, and thin-bed tuning effects. To estimate this time-variant phase, higher-order statistical measures—specifically kurtosis and skewness—are traditionally maximized to drive the signal toward a maximally non-Gaussian or maximally asymmetric zero-phase state. This paper addresses the computational and stability challenges inherent in nonstationary estimation by casting the problem as a regularized non-convex optimization task. We propose a robust framework based on the Alternating Direction Method of Multipliers (ADMM) that eliminates the instability and artifacts associated with traditional piecewise-stationary windowed approaches. The core of our contribution is the derivation of the first closed-form proximity operators for the scale-invariant inverse kurtosis ($\ell_2^4/\ell_4^4$) and inverse skewness ($\ell_2^3/\ell_3^3$) functionals. By exploiting the signed permutation invariance of these statistical measures, we reduce the high-dimensional proximal subproblems to efficient one-dimensional root-finding tasks. We provide a detailed geometric interpretation of the optimality conditions, demonstrating that the global minimizer is governed by a branch-separation property. Furthermore, we derive an explicit critical threshold parameter, $\mu_c$, which provides a theoretical rule for identifying the global minimum among multiple stationary points. Numerical validations on synthetic and real seismic data demonstrate that the proposed proximal algorithms achieve $O(n)$ computational complexity and superior stability compared to traditional methods, effectively enabling high-resolution, nonstationary phase correction in complex geological environments.
\end{abstract}


\section{Introduction}

The seismic trace $s(t)$ is traditionally modeled as the convolution of a seismic wavelet $w(t)$ and the Earth's reflectivity series $r(t)$ \cite{Wiggins_1978_MED}. In seismic data processing, the estimation and correction of the wavelet phase is a critical task; while physical causal systems are inherently minimum-phase, seismic interpretation relies on zero-phase wavelets to maximize temporal resolution and ensure that reflection times correspond accurately to subsurface interfaces
\cite{Robinson_2000_GSA,Schoenberger_1974_RCM}. However, the wavelet phase is frequently nonstationary, undergoing spatial and temporal variations due to physical phenomena such as attenuation, dispersion, and thin-bed tuning effects 
\cite{vanderBan_2009_NPE,vanderBan_2008_TVW}.

Statistical methods for phase estimation generally rely on higher-order statistics (HOS) under the assumption that the underlying reflectivity is non-Gaussian and white \cite{Wiggins_1978_MED}. The central limit theorem provides the mathematical rationale: the convolution of a filter with a non-Gaussian white series renders the result more Gaussian; thus, the optimal zero-phase state is identified by maximizing the non-Gaussianity of the signal \cite{Donoho_1981_OME, vanderBan_2009_NPE}. Historically, the varimax norm, or kurtosis (a fourth-order statistic), has been the primary objective function for this purpose \cite{Levy_1987_APC, Wiggins_1978_MED}. More recently, skewness (a third-order statistic) has been proposed as a superior attribute due to its higher dynamical range and better stability in detecting the asymmetry inherent in Earth's reflectivity distributions \cite{vanderBan_2008_TVW,vanderBan_2009_NPE,Fomel_2014_LSA}.

To address the time-varying nature of the wavelet, traditional approaches subdivide the seismic section into overlapping time windows, assuming piecewise stationarity within each window. However, these methods are highly sensitive to the chosen window length and often produce unstable, oscillatory phase estimates if the windows are too short \cite{vanderBan_2008_TVW,vanderBan_2009_NPE}. Recent advancements have moved toward recast-ing the problem as a regularized least-squares optimization to produce smoothly varying local attributes \cite{vanderBan_2008_TVW,vanderBan_2009_NPE,Fomel_2014_LSA}.

Despite their effectiveness, the objective functions involved—specifically the scale-invariant inverse kurtosis ($\ell_2^4/\ell_4^4$) and inverse skewness ($\ell_2^3/\ell_3^3$)—are non-convex and highly nonlinear, posing significant computational challenges in high-dimensional settings. In this paper, we propose a robust numerical framework based on the Alternating Direction Method of Multipliers (ADMM) to solve the nonstationary phase estimation problem \cite{Powell_1969_NLC,Nocedal_2006_NO,Boyd_2011_DOS} . The core contribution of this work is the derivation of the first closed-form proximity operators for these scale-invariant statistical measures. 
We reduce the high-dimensional proximal subproblems to efficient one-dimensional root-finding tasks. This approach achieves $O(n)$ computational complexity, providing a stable and efficient alternative to traditional windowed or grid-search methods for high-resolution seismic processing.

The paper is organized as follows. Section 2 formulates the phase correction problem within the ADMM framework. Sections 3 and 4 present the analytical derivation of the proximity operators for inverse kurtosis and inverse skewness, respectively. Section 5 provides numerical validation of the proposed algorithms on synthetic and real seismic data.

\section{Mathematical Formulation and Optimization}

In this section, we formulate the nonstationary phase estimation problem as a regularized non-convex optimization task and propose an efficient numerical scheme based on the Alternating Direction Method of Multipliers (ADMM) \cite{Nocedal_2006_NO}.

\subsection{Problem Statement and Objective Functions}
The time-varying phase $\phi$ of a discrete seismic signal $s$ is estimated by maximizing a statistical measure of non-Gaussianity $\kappa(\cdot)$ applied to the phase-rotated signal $s_{\text{rot}}(\phi)$ \cite{vanderBan_2008_TVW}
\begin{equation}
\phi = \arg\max_{\phi} \kappa(s_{\text{rot}}(\phi)).
\end{equation}
The rotated trace is computed using the original trace $s$ and its Hilbert transform $H[s]$ such that $s_{\text{rot}}(\phi) = s \cos \phi +H[s] \sin \phi$ \cite{Levy_1987_APC}. We specifically consider the scale-invariant kurtosis and skewness functionals:
\begin{align}
\text{Kurtosis: } \kappa(s_{\text{rot}}) &= \frac{\sum_i s_{\text{rot},i}^4(\phi)}{\left( \sum_i s_{\text{rot},i}^2(\phi) \right)^2} = \frac{\|s_{\text{rot}}\|_4^4}{\|s_{\text{rot}}\|_2^4}, \\
\text{Skewness: } \kappa(s_{\text{rot}}) &= \frac{\sum_i s_{\text{rot},i}^3(\phi)}{\left( \sum_i s_{\text{rot},i}^2(\phi) \right)^{3/2}} = \frac{\|s_{\text{rot}}\|_3^3}{\|s_{\text{rot}}\|_2^3}.
\end{align}
These optimization problems are highly nonlinear and cannot be solved directly using standard linear techniques \cite{Aster_2004_PEI}. Since $\kappa(s_{\text{rot}}) > 0$, maximizing the statistical measure is equivalent to minimizing its reciprocal. 

\subsection{Regularized Framework and Variable Splitting}
To overcome the instability of traditional windowed approaches, we introduce a variable-splitting strategy  \cite{Beck_2009_FIS}. By defining an auxiliary variable $x = s_{\text{rot}}(\phi)$, we recast the problem as the following constrained minimization:
\begin{equation}
\minimize_{\phi, x}~ \kappa^{-1}(x) + \alpha R(\phi) \quad \text{subject to } x = s_{\text{rot}}(\phi),
\end{equation}
where $R(\phi)$ is a regularization term (e.g., a discrete Sobolev or Tikhonov norm) promoting temporal and spatial smoothness of the estimated phase function \cite{vanderBan_2009_NPE,Fomel_2014_LSA}. The parameter $\alpha > 0$ balances this smoothness against the desired non-Gaussianity of the rotated signal.

\subsection{ADMM Iterative Scheme}
The associated Augmented Lagrangian for the constrained problem is given by  \cite{Powell_1969_NLC,Nocedal_2006_NO}:
\begin{equation}
\mathcal{L}(x, \phi, \lambda; \mu) = \kappa^{-1}(x) + \alpha R(\phi) + \lambda^T (x - s_{\text{rot}}(\phi)) + \frac{\mu}{2} \|x - s_{\text{rot}}(\phi)\|_2^2,
\end{equation}
where $\lambda$ is the Lagrange multiplier and $\mu > 0$ is the penalty parameter. Applying the ADMM framework yields the following iterative scheme \cite{Boyd_2011_DOS}:
\begin{subequations}
\label{ADMM_phase}
\begin{align}
x^{k+1} &= \arg\min_{x} \kappa^{-1}(x) + \frac{\mu}{2} \|x - s_{\text{rot}}(\phi^k) + \lambda^k\|_2^2, \label{subeq:x_update} \\
\phi^{k+1} &= \arg\min_{\phi} \frac{\mu}{2} \|x^{k+1} - s_{\text{rot}}(\phi) + \lambda^k\|_2^2 + \alpha R(\phi), \label{subeq:phi_update} \\
\lambda^{k+1} &= \lambda^k + x^{k+1} - s_{\text{rot}}(\phi^{k+1}). \label{subeq:lambda_update}
\end{align}
\end{subequations}

The efficiency of the proposed scheme depends on the resolution of the subproblems:
\begin{enumerate}
    \item \textbf{The $x$-update (\ref{subeq:x_update}):} This step represents the application of a proximity operator for the inverse kurtosis or inverse skewness. In Sections  \ref{prox_kustosis} and \ref{prox_skewness}, we derive the first closed-form solutions for these operators, reducing high-dimensional minimizations to efficient $O(n)$ root-finding tasks.
    \item \textbf{The $\phi$-update (\ref{subeq:phi_update}):} This subproblem is nonlinear but can be solved efficiently via a single Gauss--Newton iteration \cite{Tapia_1977_DMM}. Defining the residual $r(\phi) = x^{k+1} - s_{\text{rot}}(\phi) + \lambda^k$ and letting $J(\phi) = \frac{\partial s_{\text{rot}}(\phi)}{\partial \phi}$ denote the Jacobian, the update $\delta\phi$ is obtained by solving \cite{Nocedal_2006_NO}
    \begin{equation}
    (\mu J_k^T J_k + \alpha \nabla^2 R(\phi^k)) \delta\phi = \mu J_k^T r_k - \alpha \nabla R(\phi^k).
    \end{equation}
    Crucially, because phase rotation is a pointwise operation, the Jacobian $J(\phi)$ is diagonal. This results in a diagonal $J_k^T J_k$ term, allowing the data misfit contribution to decouple across samples. Coupling between samples is introduced solely through the regularization term $R(\phi)$, making the update computationally stable and efficient.
\end{enumerate}

\section{Proximity Operator of the $\ell_2^4/\ell_4^4$ Ratio} \label{prox_kustosis}

In this section, we investigate the computation of the proximity operator associated with the nonconvex, scale-invariant inverse kurtosis functional $\kappa^{-1}: \mathbb{R}^n \setminus \{0\} \to \mathbb{R}$, defined as the fourth power of the ratio between the $\ell_2$ and $\ell_4$ norms:
\begin{equation} \label{kurt0}
\kappa^{-1}(x) = \frac{\|x\|_2^4}{\|x\|_4^4}.
\end{equation}
For a given input vector $y \in \mathbb{R}^n$ and a proximal parameter $\mu > 0$, the proximity operator $\text{prox}_{\mu \kappa^{-1}}(y)$ is defined as the solution to the following minimization problem:
\begin{equation}
\text{prox}_{\mu \kappa^{-1}}(y) = \arg\min_{x \in \mathbb{R}^n} \left\{ \Phi(x) := \frac{1}{2}\|x - y\|_2^2 + \mu \frac{\|x\|_2^4}{\|x\|_4^4} \right\}.
\end{equation}
The function $\kappa^{-1}(x)$ is scale-invariant, satisfying $h(c\, x) = h(x)$ for any $c \neq 0$. While this property is ideal for promoting spikeness, it renders the objective function $\Phi(x)$ nonconvex.

To simplify the computation, we exploit the signed permutation invariance of the $\ell_2$ and $\ell_4$ norms. As established in Theorem \ref{theorem1}, the minimizer $x^*$ must share the same sign pattern as the input vector $y$, satisfying $\text{sign}(x^*) = \text{sign}(y)$. Furthermore, if the components of $y$ are sorted in non-decreasing order ($0 \le y_{(1)} \le y_{(2)} \le \dots \le y_{(n)}$), there exists a global minimizer $x^*$ that preserves this monotonicity:
\begin{equation}
0 \le x^*_{(1)} \le x^*_{(2)} \le \dots \le x^*_{(n)}.
\end{equation}
Consequently, we restrict our analysis to the nonnegative orthant and assume $y$ is pre-sorted in ascending order.

\begin{theorem}[Existence of an ordered proximal minimizer]\label{theorem1}
Let $\mu>0$ and $y\in\mathbb{R}^n$. Consider the proximal problem
\[
\min_{x\in\mathbb{R}^n}
\Bigg\{
\Phi(x)
:=
\frac12\|x-y\|_2^2
+\mu\,\frac{\|x\|_2^4}{\|x\|_4^4}
\Bigg\}.
\]
Assume that $y$ is nonnegative and sorted in ascending order,
\[
0\le y_{(1)}\le y_{(2)}\le \cdots \le y_{(n)}.
\]
Then there exists at least one global minimizer $x^\star$ of $\Phi$ such that
\[
0\le x_{(1)}^\star\le x_{(2)}^\star\le \cdots \le x_{(n)}^\star.
\]
\end{theorem}
\begin{proof}
Define the regularizer $\kappa^{-1}(x)$ as in \eqref{kurt0}.
Since both $\|\cdot\|_2$ and $\|\cdot\|_4$ are invariant under permutations of
the components, $\kappa^{-1}(x)$ is permutation invariant, i.e.,
\[
\kappa^{-1}(\Pi x)=\kappa^{-1}(x)
\qquad \text{for any permutation matrix } \Pi.
\]

Let $x$ be any global minimizer of $\Phi$. Suppose $x$ is not sorted. Then
there exist indices $i<j$ such that $x_i>x_j$. Since $y$ is sorted, we also
have $y_{(i)}\le y_{(j)}$. Let $\tilde x$ be obtained by swapping $x_i$ and $x_j$.
Then $\kappa^{-1}(\tilde x)=\kappa^{-1}(x)$. Moreover, a direct expansion yields
\[
(x_i-y_{(i)})^2+(x_j-y_{(j)})^2
-
\Big[(x_j-y_{(i)})^2+(x_i-y_{(j)})^2\Big]
=
2(x_i-x_j)(y_{(j)}-y_{(i)})\ge 0,
\]
which implies $\|\tilde x-y\|_2^2\le \|x-y\|_2^2$. Therefore, $\Phi(\tilde x) \le \Phi(x)$.
Since $x$ is a minimizer, $\tilde x$ is also a minimizer. Repeating this
exchange argument finitely many times eliminates all inversions and produces
a minimizer $x^\star$ satisfying
\[
x_{(1)}^\star\le x_{(2)}^\star\le \cdots \le x_{(n)}^\star.
\]
\end{proof}

\subsection{First-Order Optimality Conditions and Auxiliary Parameters}
Applying the first-order necessary conditions for optimality to $\Phi(x)$ for $x \neq 0$ yields a system of nonlinear equations for each component $x_{(i)}$:
\begin{equation}
x_{(i)} - y_{(i)} + 4\mu \frac{\|x\|_2^2}{\|x\|_4^4}\left( 1-\frac{\|x\|_2^2 }{\|x\|_4^4} x_{(i)}^2\right) x_{(i)} = 0.
\end{equation}
To analyze this system, we define the auxiliary scalar parameter $\alpha$ based on the norm ratio of the solution:
\begin{equation} \label{alpha}
\alpha = \frac{\|x\|_2^2}{\|x\|_4^4}.
\end{equation}
In terms of this parameter, the optimality conditions reduce to a family of cubic equations 
\begin{equation}\label{xi_cube}
x_{(i)}^3 - px_{(i)} + q_{(i)} = 0,
\end{equation}
where the coefficients $p$ and $q_{(i)}$ are defined as:
\begin{equation}\label{pq_kurt}
p = \frac{1 + 4\mu\alpha}{4\mu\alpha^2}, \quad q_{(i)} = \frac{y_{(i)}}{4\mu\alpha^2}.
\end{equation}

\subsection{Geometric Interpretation}

The coordinate-wise cubic optimality conditions in \eqref{xi_cube} can be equivalently expressed as the intersection of two basic functions using the classical method of Omar Khayyam \cite{Vali_2021_OKG}:
\begin{equation}
x_{(i)}^2 - p = -\frac{q_{(i)}}{x_{(i)}}.
\end{equation}
This formulation allows us to interpret each solution component $x_{(i)}$ as the intersection between two scalar functions: a fixed convex parabola $P(x) = x^2 - p$ and a hyperbolic $H_i(x) = -q_{(i)}/x$. The parameter $q_{(i)}$, which is directly proportional to the input magnitude $y_{(i)}$, dictates the shape of the hyperbola for each coordinate.

As illustrated in Figure \ref{cubic_roots_kurt}, the intersection points of these two functions in the positive orthant define the potential candidates for the proximal solution. For a given index $i$, the cubic nature of the optimality equation yields three real roots, exactly two of which are positive: a small-root branch $r_i$ and a large-root branch $R_i$. 
The figure shows the dynamic behavior of these roots as the input values increase. As $y_{(i)}$ (and thus $q_{(i)}$) increases, the hyperbolic curve $H_i(x)$ shifts further from the axes, causing the small roots $r_i$ to move monotonically to the right and the large roots $R_i$ to move monotonically to the left. This visual movement confirms the monotonicity properties established in Theorem \ref{theorem2}:
\begin{equation}\label{rR}
r_1 < r_2 < \dots < r_n \quad \text{and} \quad R_1 > R_2 > \dots > R_n.
\end{equation}

A critical insight provided by Figure \ref{cubic_roots_kurt} is the strict separation between these two solution branches. Because the hyperbolic curves $H_i(x)$ are nested within each other while the parabola $P(x)$ remains fixed, the smallest element of the large-root branch $R_n$ is guaranteed to be strictly greater than every element of the small-root branch. 
This separation is vital for the branch selection rule established in Theorem \ref{theorem3}. It ensures that an ordered proximal solution must utilize the small-root branch for the first $n-1$ components to maintain monotonicity. Consequently, the decision-making process for the proximity operator is simplified to determining whether the leading component $x_{(n)}$ should remain on the small-root branch or switch to the large-root branch as the regularization parameter $\mu$ increases.

\begin{figure}[!h]
\center
\includegraphics[scale=1.5]{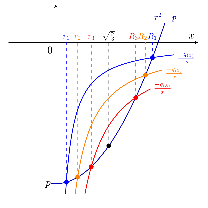}
\caption{Geometric interpretation of the cubic optimality condition. The parabolic function $P(x)=x^2-p$ and the hyperbolic functions $H_i(x)=-q_{(i)}/x$ are shown for three increasing values of $q_{(i)}$. For each $i$, their intersection points define the two real roots $r_i$ (small-root branch) and $R_i$ (large-root branch) of the equation 
$x_{(i)}^3-px+q_{(i)}=0$. As $q_{(i)}$ increases, the small roots move monotonically to the right, while the large roots move monotonically to the left.}
\label{cubic_roots_kurt}
\end{figure}

\begin{theorem}[Ordering of the two positive root branches] \label{theorem2}
Let $p>0$ be fixed and consider, for each $i=1,\dots,n$, the depressed cubic
$x^3 - px + q_{(i)} = 0$,
where $q_i>0$ and the discriminant $\frac{q(i)^2}{4} + \frac{p^3}{27}$ is negative so that the equation admits
three distinct real roots, exactly two of which are positive. Denote these
two positive roots by $0<r_i<R_i$.
Assume that $y\in\mathbb{R}^n$ is nonnegative and sorted,
$0<y_{(1)} \le y_{(2)} \le \cdots \le y_{(n)}$,
and that $q_{(i)}$ is defined in \eqref{pq_kurt}.
Then the vectors of positive roots satisfy the monotonicity properties \eqref{rR}.
Moreover, since $r_n < R_n$, it follows that
\[
R_n > \max_{1 \le i \le n} r_i.
\]
\end{theorem}

\begin{proof}
Define $F(x,q)=x^3 - px + q$. Any root $x(q)$ satisfies $F(x(q),q)=0$.
By implicit differentiation,
\begin{equation}
\frac{dx}{dq}
=
-\frac{1}{3x^2 - p}.
\end{equation}
%
Since the discriminant is negative, the equation admits three distinct
real roots, two of which are positive. The smaller positive root $r(q)$
lies in $(0,\sqrt{p/3})$, hence $3r(q)^2 - p < 0$, and therefore $\frac{dr}{dq} > 0$.
Thus, $r(q)$ is strictly increasing as a function of $q$.
From the definition $
q_{(i)} = \frac{y_{(i)}}{4\mu\alpha^2}$,
and the ordering $y_{(1)} \le \cdots \le y_{(n)}$, it follows that $q_{(1)} \le q_{(2)} \le \cdots \le q_{(n)}$.
Since $r(q)$ is increasing in $q$, we conclude that
$r_1 < r_2 < \cdots < r_n$.

For the larger positive root $R(q)$, we have $R(q) > \sqrt{p/3}$, hence
$3R(q)^2 - p > 0$, which implies $\frac{dR}{dq} < 0$.
Thus, $R(q)$ is strictly decreasing as a function of $q$. Therefore,
since $q_{(i)}$ increases with $i$, it follows that $R_1 > R_2 > \cdots > R_n$.
Finally, since $r_i < R_i$ for all $i$ and $\{r_i\}$ is increasing, we have
\[
r_i \le r_n < R_n, \qquad i=1,\dots,n,
\]
which implies
\[
R_n > \max_{1 \le i \le n} r_i.
\]
\end{proof}

\begin{theorem}[Small-root selection for the first $n-1$ components] \label{theorem3}
Assume that $y\in\mathbb{R}^n$ is nonnegative and sorted,
$0<y_{(1)}\le y_{(2)}\le \cdots \le y_{(n)}$.
For each $i=1,\dots,n$, consider the cubic optimality equation
\begin{equation}
4\mu\alpha^2 x_i^3-(1+4\mu\alpha)x_i+y_{(i)}=0,
\end{equation}
and assume that it admits exactly two positive roots $0<r_i<R_i$.
Suppose further that the roots satisfy \eqref{rR} and that
$R_n>r_{n-1}$.
Then any ordered minimizer $x^\star$ of the proximal problem satisfies
\[
x_i^\star=r_i,
\qquad i=1,\dots,n-1.
\]
In other words, the first $n-1$ components must coincide with the smaller
positive roots.
\end{theorem}

\begin{proof}
Let $x^\star$ be an ordered minimizer, so that
$0\le x_1^\star\le x_2^\star\le \cdots \le x_n^\star$. Since $x_i^\star$ must satisfy the cubic equation,
we have either $x_i^\star=r_i$ or $x_i^\star=R_i$.

Assume for contradiction that $x_i^\star=R_i$ for some $i\le n-1$.
Because the sequence $\{R_i\}$ is strictly decreasing,
$R_i\ge R_{n-1}>R_n$.
On the other hand, by ordering we must have $x_i^\star\le x_n^\star$.
The largest admissible value that $x_n^\star$ can take among the stationary
roots is $R_n$, hence $x_n^\star\le R_n$. Therefore,
$R_i=x_i^\star\le x_n^\star\le R_n$,
which contradicts $R_i>R_n$. Hence $x_i^\star\neq R_i$ for all $i\le n-1$.
Consequently,
\[
x_i^\star=r_i,
\qquad i=1,\dots,n-1,
\]
which proves the claim.
\end{proof}

\subsection{The Feasibility Region}

The existence of real-valued, positive solution candidates for the coordinate-wise cubic equations derived in \eqref{xi_cube} depends on the discriminant of the cubic form. For a nonzero proximal solution to exist, the parameters must satisfy a condition that ensures the cubic $x^3 - px + q_{(i)} = 0$ admits positive roots. Since $y_{(n)}$ is the largest component of the sorted input, the global feasibility condition is determined by the $n$-th coordinate:
\begin{equation}
\frac{y_{(n)}^2}{(4\mu\alpha^2)^2} \le \frac{4(1 + 4\mu\alpha)^3}{27(4\mu\alpha^2)^3}.
\end{equation}
By simplifying this expression and substituting the definitions of the auxiliary parameters, we obtain the fundamental feasibility inequality:
\begin{equation} \label{feas_ineq}
(1 + 4\mu\alpha)^3 \ge 27\mu\alpha^2 y_{(n)}^2.
\end{equation}

The boundary of the feasibility region is defined by the equality $(1 + 4\mu\alpha)^3 = 27\mu\alpha^2 y_{(n)}^2$ and equivalently the explicit parametric solution is:
\begin{equation} \label{feas_par}
\mu(u) = \frac{27u^2y_{(n)}^2}{(1+4u)^3}, \quad
\alpha(u) = \frac{(1+4u)^3}{27uy_{(n)}^2}, \quad u>0.
\end{equation}

Figure \ref{cubic_feas} illustrates the feasible and infeasible regions in the $(\mu, \alpha)$ plane, where the black curve indicates the boundary. As established by the boundary equation, the infeasible set is bounded from below by the threshold $\alpha = 1/y_{(n)}^2$ and from the right by the $\mu = y_{(n)}^2/4$.

\begin{figure}[!ht]
\center
\includegraphics[scale=1.5]{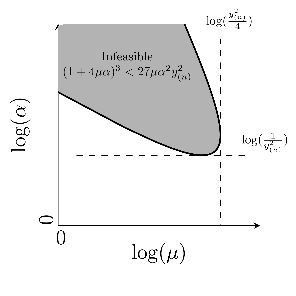}
\caption{The feasibility region for the $\ell_2^4/\ell_4^4$ ratio in the $(\mu, \alpha)$ plane. The red boundary curve separates the feasible and infeasible zones for the existence of real-valued solution branches. The infeasible region is bounded from below by the physical threshold $\alpha = 1/y_{(n)}^2$  and from the right by the $\mu = y_{(n)}^2/4$.}
\label{cubic_feas}
\end{figure}

\subsection{Trigonometric Form of the Solution Branches}
Utilizing Cardano's trigonometric form for depressed cubics \cite{Zucker_2008_TCE}, the cubic root branches are expressed as:
\begin{equation}\label{r_R_cubic}
r_i = 2\sqrt{\frac{p}{3}}\cos\Big(\frac{\theta_i+4\pi}{3}\Big), \quad R_i = 2\sqrt{\frac{p}{3}}\cos\Big(\frac{\theta_i}{3}\Big)
\end{equation}
where the auxiliary angle $\theta_i$ is dictated by the ratio of the coefficients:
\begin{equation}
\cos\theta_i = -\frac{3\sqrt{3} q_{(i)}}{2\sqrt{p^3}}.
\end{equation}
%
%
%

As established in Theorem \ref{theorem3}, the proximal solution for the $\ell_2^4/\ell_4^4$ ratio maintains a fixed structural pattern: the first $n-1$ coordinates reside on the small-root branch $r_i$, while the largest coordinate $x_{(n)}$ may switch from $r_n$ to the large-root branch $R_n$ as the regularization parameter $\mu$ increases. The transition between these regimes occurs at a unique critical value denoted as $\mu_c$.

\subsection{Determination of the Critical Parameter $\mu_c$}

At the critical point $\mu = \mu_c$, the two positive solution candidates for the largest component coincide, i.e., $r_n = R_n$. In terms of the cubic optimality condition, this corresponds to the vanishing of the coordinate-wise discriminant, leading to the following boundary condition for the internal norm ratio $\alpha$ and the input magnitude $y_{(n)}$:
\begin{equation} \label{Feas}
(1 + 4\mu\alpha)^3 = 27\mu\alpha^2 y_{(n)}^2.
\end{equation}
In the trigonometric representation, this vanishing discriminant is equivalent to the condition $\theta_n = \pi$. Under this limit, the auxiliary angles for all other coordinates are dictated by the relative magnitudes of the input vector:
\begin{equation}
\cos \theta_i = -\frac{y_{(i)}}{y_{(n)}}.
\end{equation}
At the transition point, the solution components $x_{(i)}$ are defined entirely by the small-root branch, which can be simplified using the critical angles. To facilitate a compact analytical expression for $\mu_c$, we introduce an auxiliary vector $v \in \mathbb{R}^n$ with components defined as:
\begin{equation}
v_i = \cos \left( \frac{\arccos(-y_{(i)}/y_{(n)}) + 4\pi}{3} \right).
\end{equation}
Using these components to define the solution $x_{(i)}:=r_{(i)}$ \eqref{r_R_cubic} and solving the two coupled systems \eqref{alpha} and \eqref{Feas} yields the explicit formula for the critical parameter:
\begin{equation} \label{mu_c_kurt}
\mu_c = \frac{y_{(n)}^2 \left( \sum_{i=1}^n v_i^4 \right)^2\left( 3 \sum_{i=1}^n v_i^2 - 4 \sum_{i=1}^n v_i^4 \right)}{\left( \sum_{i=1}^n v_i^2 \right)^3 }.
\end{equation}
This closed-form expression is a fundamental component of the proximity operator, providing the exact threshold for the selection rule that determines the global minimizer.

\subsection{Determination of the Internal Parameter $\alpha$}

For a given proximal parameter $\mu$, the internal norm ratio $\alpha = \|x\|_2^2 / \|x\|_4^4$ must satisfy a self-consistency condition where the $\alpha$ used to define the cubic coefficients $p$ and $q_{(i)}$ matches the norm ratio of the resulting solution vector $x$. We reduce this problem to a one-dimensional root-finding search for $\alpha$.

\subsubsection{Computational Procedure}
The following procedure ensures that the proximity operator is evaluated efficiently and that the solution remains within the feasibility region:

\begin{enumerate}
    \item Compute the critical parameter $\mu_c$ using the closed-form expression derived in \eqref{mu_c_kurt}.
    \item Based on the relationship between $\mu$ and $\mu_c$, determine the branch assignment for the largest coordinate $x_{(n)}$:
    \begin{itemize}
        \item If $\mu \le \mu_c$, assign $x_{(n)}$ to the small-root branch $r_n$.
        \item If $\mu > \mu_c$, assign $x_{(n)}$ to the large-root branch $R_n$.
    \end{itemize}
    All other components $x_{(i)}$ for $i = 1, \dots, n-1$ are assigned to the small-root branch $r_i$ .
    \item Solve for the root $\alpha^*$ of the residual function:
    \begin{equation}
    f(\alpha) = \alpha - \frac{\sum_{i=1}^n x_{(i)}(\alpha)^2}{\sum_{i=1}^n x_{(i)}(\alpha)^4} = 0,
    \end{equation}
    where $x_{(i)}(\alpha)$ are the trigonometric roots defined in \eqref{r_R_cubic}.
\end{enumerate}

The search space for $\alpha$ is governed by the global feasibility condition established: $(1 + 4\mu\alpha)^3 \ge 27\mu\alpha^2 y_{(n)}^2$. 
Depending on the magnitude of $\mu$ relative to the input magnitude $y_{(n)}$, we identify two distinct regimes for the root-finding process:
\begin{itemize}
    \item If $\mu > y_{(n)}^2/4$, the discriminant of the feasibility inequality remains positive for all $\alpha > 0$. In this case, no infeasibility occurs, and the root-finding problem is well-defined over the entire positive domain $(0, \infty)$.
    \item If $\mu \le y_{(n)}^2/4$, the boundary equation $(1 + 4\mu\alpha)^3 = 27\mu\alpha^2 y_{(n)}^2$ admits two real roots for $\alpha$, denoted as $\alpha_l$ (lower) and $\alpha_u$ (upper) (Figure \ref{cubic_feas}). These roots define two disjoint feasible intervals: $(0, \alpha_l]$ and $[\alpha_u, \infty)$.
To determine which interval contains the global minimizer, we examine the behavior of the proximity operator as the regularization vanishes. 
As $\mu \to 0$, the proximity operator approaches the identity operator, meaning $x \to y$. Consequently, the optimal norm ratio $\alpha^*$ must approach the norm ratio of the input vector:
    \begin{equation}
    \lim_{\mu \to 0} \alpha^*(\mu) = \alpha_y =\frac{\|y\|_2^2}{\|y\|_4^4}.
    \end{equation}
However, in this limit, as seen form \eqref{feas_par} $\alpha_l$ tends toward infinity. Thus, it is guaranteed that for sufficiently small $\mu$, the physical solution $\alpha$ falls within the lower interval $(0, \alpha_l]$. 
On the other hand, since the proximity operator $\text{prox}_{\mu \kappa^{-1}}(y)$ is a continuous mapping with respect to $\mu$ (outside of the point where it might jump to the origin), the optimal parameter $\alpha^*(\mu)$ must follow a continuous path starting from $\alpha_y$ at $\mu=0$. Because the two feasible intervals $(0, \alpha_l]$ and $[\alpha_u, \infty)$ are disjoint, we do not expect $\alpha^*$ to ``jump" into the upper interval without first crossing the infeasible region, which is mathematically impossible for a stationary point branch. Consequently, we restrict the numerical search for $\alpha^*$ to this sub-interval to ensure convergence to the global minimizer while maintaining computational efficiency. Once $\alpha^*$ is identified, the final solution vector $x^*$ is reconstructed using the chosen branches.
\end{itemize}

Figure \ref{L2L4_toy} illustrates the evolution of $\alpha$ as a function of $\mu$ for the input vector $y = [1,\,2,\,3]^T$. The infeasible region is highlighted in blue. The blue markers denote the corresponding values of $\alpha$, which approach the limit $\|y\|_2^2 / \|y\|_4^4 \approx 0.1429$ as $\mu \to 0$. As $\mu$ increases, $\alpha$ decreases and reaches the boundary of the infeasible region (shown by the red curve) at the critical threshold $\mu_c \approx 0.83$. Beyond this point, $\alpha$ continues to decrease, attaining its minimum value of approximately $0.0943$ at $\mu \approx 5.2811$. For larger values of $\mu$, $\alpha$ increases again and asymptotically approaches $1/3^2$.
We note that, in the limit $\mu \to \infty$, the solution becomes maximally sparse: only the largest component is retained, with $x_{(n)} = y_{(n)}$, while all other entries vanish.

\begin{figure}[!ht]
\center
\includegraphics[scale=0.45]{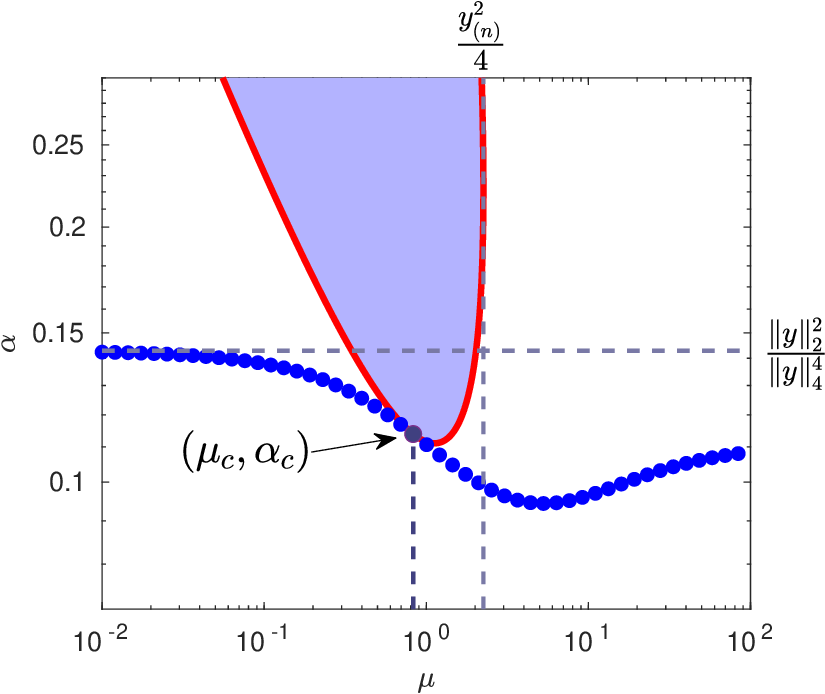}
\hspace{1cm}
\includegraphics[scale=0.45]{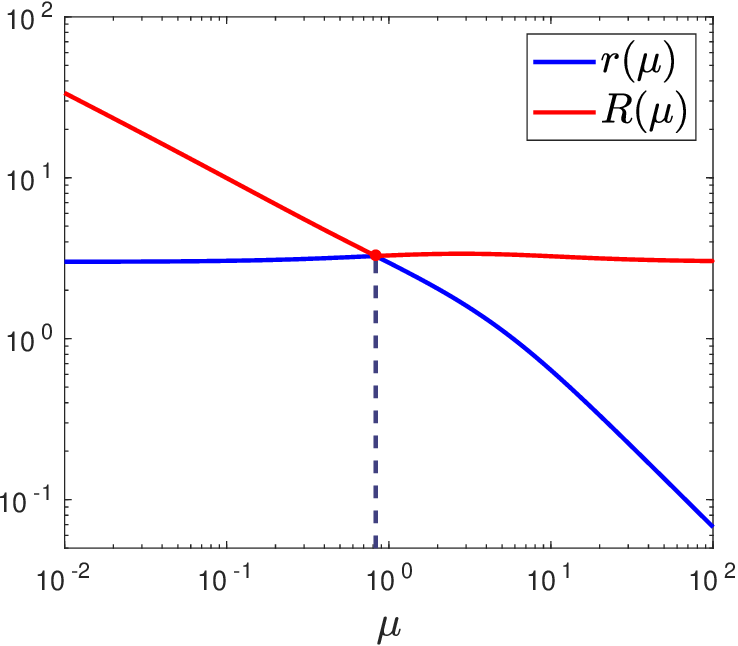}
\caption{Left: the evolution of $\alpha$ as a function of $\mu$ for the proximity operator of inverse kurtosis with input vector $y = [1,\,2,\,3]^T$.
Right: the small and large root branches corresponding to the larger sample 3.
}
\label{L2L4_toy}
\end{figure}

\section{Proximity Operator of the $\ell_2^3/\ell_3^3$ Ratio} \label{prox_skewness}
In this section, we investigate the computation of the proximity operator associated with the inverse skewness, defined as the cube of the ratio between the $\ell_2$ and $\ell_3$ norms:
\begin{equation}
\text{prox}_{\mu \kappa^{-1}}(y) = \arg\min_{x \in \mathbb{R}^n} \left\{ \Phi(x) := \frac{1}{2}\|x - y\|_2^2 + \mu \frac{\|x\|_2^3}{\|x\|_3^3} \right\}.
\end{equation}

Similar to the inverse kurtosis function, if the components of $y$ are sorted in non-decreasing order such that $0 \le y_{(1)} \le y_{(2)} \le \dots \le y_{(n)}$, there exists at least one global minimizer $x^*$ that preserves this monotonicity:
\begin{equation}
0 \le x^*_{(1)} \le x^*_{(2)} \le \dots \le x^*_{(n)}.
\end{equation}
Consequently, without loss of generality, we restrict our analysis to the nonnegative orthant and assume the input vector $y$ is pre-sorted in ascending order.

\subsection{First-Order Optimality Conditions}
Applying the first-order necessary conditions for optimality to the objective $\Phi(x)$ for $x \neq 0$ yields a system of nonlinear equations for each component $x_{(i)}$:
\begin{equation}
x_{(i)}^2 - \left( \frac{\|x\|_3^6}{3\mu\|x\|_2^3} + \frac{\|x\|_3^3}{\|x\|_2^2} \right) x_{(i)} + \frac{\|x\|_3^6}{3\mu\|x\|_2^3} y_{(i)} = 0.
\end{equation}
To analyze this system, we define the auxiliary scalar parameters $\alpha$ and $\beta$ as follows:
\begin{equation}\label{alpha_beta}
\alpha = \frac{\|x\|_2^2}{\|x\|_3^3}, \quad \beta = \|x\|_2.
\end{equation}
In terms of these parameters, the optimality conditions reduce to a family of quadratic equations $x_{(i)}^2 - px_{(i)} + q_{(i)} = 0$, where the coefficients $p$ and $q_{(i)}$ are defined as:
\begin{equation} \label{pq}
p = \frac{\beta + 3\mu\alpha}{3\mu\alpha^2}, \quad q_{(i)} = \frac{\beta y_{(i)}}{3\mu\alpha^2}.
\end{equation}
Each component of the proximal solution must therefore satisfy 
\begin{equation} \label{xi}
x_{(i)} = \frac{p}{2} \pm \frac{1}{2}\sqrt{p^2 - 4q_{(i)}}.
\end{equation}

\subsection{Geometric Interpretation}

The coordinate-wise optimality conditions derived in the previous section can be equivalently expressed as:
\begin{equation}
x_{(i)} - p = -\frac{q_{(i)}}{x_{(i)}},
\end{equation}
where $p$ and $q_{(i)}$ are defined as in \eqref{pq}. This formulation allows us to interpret the solution $x_{(i)}$ as the intersection between two scalar functions: a linear function $L(x) = x - p$ and a family of hyperbolic functions $H_i(x) = -q_{(i)}/x$. The two intersection points define two potential solution branches for each coordinate.

\subsection{Properties of the Solution Branches}
Assuming the feasibility condition (nonnegative discriminant) holds, the quadratic equation for each component yields two distinct positive roots, which we denote as the small-root branch $r_i$ and the large-root branch $R_i$:
\begin{equation} \label{roots}
r_i = \frac{p - \sqrt{p^2 - 4q_{(i)}}}{2}, \quad R_i = \frac{p + \sqrt{p^2 - 4q_{(i)}}}{2},
\end{equation}
with $0 < r_i \le R_i$. By leveraging the fact that $y$ (and consequently $q$) is sorted in non-decreasing order, we establish the following critical monotonicity and separation properties:

\begin{itemize}
    \item Monotonicity: The roots in the small-root branch are strictly increasing with respect to $q_{(i)}$, whereas the large-root branch is strictly decreasing:
    \begin{equation}
    r_1 < r_2 < \dots < r_n, \quad R_1 > R_2 > \dots > R_n.
    \end{equation}
    \item Branch Separation: Because $r_n \le R_n$, it follows that $R_n \geq \max_{1 \le i \le n} r_i$. This implies a strict separation where the smallest element of the large-root branch always exceeds every element of the small-root branch. This separation is vital for the selection rule discussed later, as the proximal solution typically utilizes the small-root branch for the first $n-1$ components and evaluates a switch to the large-root branch only for the final, largest component $x_{(n)}$.
\end{itemize}

Figure \ref{cubic_roots_Skuness} illustrates the graphs of $L(x)$ and $H_i(x)$ for three representative values $q_{(i)}$, highlighting how the intersection points determine the two admissible branches of solutions. 

\begin{figure}[!ht]
\center
\includegraphics[scale=1.4]{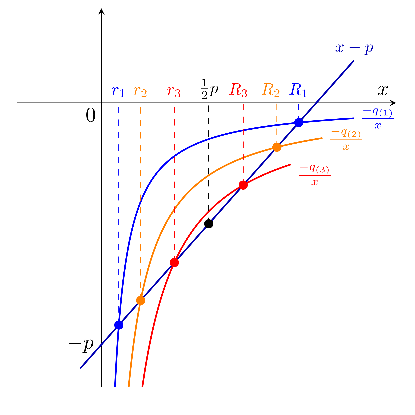}
\caption{Geometric interpretation of the quadratic optimality condition. The linear function $L(x)=x-p$ and the hyperbolic functions $H_i(x)=-q_{(i)}/x$ are shown for three increasing values of $q_{(i)}$. For each $i$, their intersection points define the two real roots $r_i$ (small-root branch) and $R_i$ (large-root branch) of the equation 
$x_{(i)}^2-px+q_{(i)}=0$. As $q_{(i)}$ increases, the small roots move monotonically to the right, while the large roots move monotonically to the left.}
\label{cubic_roots_Skuness}
\end{figure}


\subsection{The Feasibility Region}

The existence of a real-valued solution to the coordinate-wise quadratic equation \eqref{xi} depends on the nonnegativity of the discriminant. Since $y_{(n)}$ is the largest component of the sorted input vector, the global feasibility condition is determined by the $n$-th component:
\begin{equation}
\Delta = p^2 - 4q_{(n)} \ge 0.
\end{equation}
By substituting the definitions of the auxiliary parameters $p$ and $q_{(n)}$ from \eqref{pq}, we obtain a fundamental inequality that the norms of the optimal solution must satisfy:
\begin{equation}
\left( \frac{\beta + 3\mu\alpha}{3\mu\alpha^2} \right)^2 \ge \frac{4\beta y_{(n)}}{3\mu\alpha^2}.
\end{equation}
Multiplying by the positive quantity $(3\mu\alpha^2)^2$ and expanding the square leads to the following quadratic inequality in terms of the $\ell_2$-norm $\beta$:
\begin{equation}
\beta^2 + \left( 6\mu\alpha - 12\mu\alpha^2 y_{(n)} \right) \beta + 9\mu^2\alpha^2 \ge 0.
\end{equation}

\subsubsection{Existence Condition for $\alpha$}
To identify the boundaries of this region, we analyze the equality $\beta^2 + (6\mu\alpha - 12\mu\alpha^2 y_{(n)})\beta + 9\mu^2\alpha^2 = 0$. For real values of $\beta$ to exist at the boundary, the discriminant of this quadratic in $\beta$ must also be nonnegative:
\begin{equation}
\Delta_{\beta} = (6\mu\alpha - 12\mu\alpha^2 y_{(n)})^2 - 36\mu^2\alpha^2 \ge 0.
\end{equation}
Simplifying this expression yields:
\begin{equation}
144\mu^2\alpha^3 y_{(n)} \left( \alpha y_{(n)} - 1 \right) \ge 0.
\end{equation}
Given that $\mu > 0$ and the norms of a nonzero solution are positive, this condition implies a lower bound on the parameter $\alpha$:
\begin{equation}
\alpha \ge \frac{1}{y_{(n)}}.
\end{equation}

\subsubsection{Boundary Curves and Asymptotic Behavior}
The feasibility region is bounded by two curves in the $(\alpha, \beta)$ plane, representing the roots of the quadratic equation for $\beta$:
\begin{equation}
\beta_{\pm} = 
3\mu\alpha\left(\sqrt{\alpha y_{(n)}}\pm \sqrt{\alpha y_{(n)}-  1}\right)^2.
\end{equation}
These two branches, $\beta_{+}$ (upper) and $\beta_{-}$ (lower), define the valid range for the $\ell_2$-norm of the solution for any given ratio $\alpha$. 
\begin{itemize}
    \item Upper Branch: As $\alpha \to \infty$, the upper boundary $\beta_{+}$ grows quadratically with $\alpha$.
    \item Lower Branch: The lower boundary $\beta_{-}$ decreases monotonically and approaches a horizontal asymptote:
    \begin{equation}
    \lim_{\alpha \to \infty} \beta_{-} = \frac{3\mu}{y_{(n)}}.
    \end{equation}
\end{itemize}

The feasibility region represents the set of all pairs $(\alpha, \beta)$ for which the first-order optimality conditions can be satisfied by real-valued components. Any numerical procedure to solve the proximity operator must ensure that the iterative updates for $\alpha$ and $\beta$ remain within this set.

\subsection{Trigonometric Form of the Solution Branches}

Assuming the feasibility condition $p^2 - 4q_{(n)} \ge 0$ established in the previous section holds, the coordinate-wise quadratic equations admit real-valued solutions. To gain deeper insight into the complementary relationship between the solution branches and to facilitate numerical stability, we introduce a trigonometric reparameterization of the quadratic formula.

Starting from the standard algebraic solution to the coordinate-wise equation:
\begin{equation}
x_{(i)} = \frac{p}{2} \left( 1 \pm \sqrt{1 - \frac{4q_{(i)}}{p^2}} \right),
\end{equation}
the feasibility assumption ensures that the ratio $4q_{(i)}/p^2$ remains within the interval for all $i = 1, \dots, n$. This allows us to define an auxiliary angle $\theta_i \in [0, \pi/2]$ such that:
\begin{equation}
\sin^2 \theta_i = \frac{4q_{(i)}}{p^2}.
\end{equation}
Substituting this definition into the discriminant yields $\sqrt{1 - \sin^2 \theta_i} = \cos \theta_i$. Consequently, the potential solutions for each coordinate can be expressed simply as:
\begin{equation}
x_{(i)} = \frac{p}{2} (1 \pm \cos \theta_i).
\end{equation}

By applying the half-angle identities, $1 - \cos \theta = 2 \sin^2(\theta/2)$ and $1 + \cos \theta = 2 \cos^2(\theta/2)$, we obtain an explicit and structured representation for the root branches $r_i$ and $R_i$:
\begin{equation}\label{r_R}
r_i = p \sin^2 \left( \frac{\theta_i}{2} \right), \quad R_i = p \cos^2 \left( \frac{\theta_i}{2} \right).
\end{equation}
This shows that for any coordinate $i$, the sum of the two potential root candidates is always equal to the linear parameter $p$: $r_i + R_i = p$.

Similar to the inverse kurtosis case, for the inverse slewness also the proximal solution follows a distinct structural pattern: the first $n-1$ components correspond to the small-root branch $r(i)$, while the final, largest component $x_{(n)}$ may correspond to either the smaller root $r(n)$ or the larger root $R(n)$. There exists a critical value of the proximal parameter, denoted as $\mu_c$, which serves as the transition point where the global minimizer switches branches for the leading component.

To illustrate these properties, consider the input vector $y = [1;2;3]$. Figure \ref{cubic_roots_Skuness2} displays the linear function $L(x) = x - p(\mu)$ together with the hyperbolic functions $H_n(x) = -q_{(n)}(\mu)/x$ for increasing values of $\mu \in [1.6,\,8.1]$.  For each value of $\mu$, the intersection points of $L(x)$ and $H_n(x)$ define the two positive roots $r_n$ (circles) and $R_n$ (squares) of the quadratic equation $x_{(n)}^2-p(\mu)x_{(n)}+q_{(n)}(\mu)=0$. The selected solution $x_{(n)}$ is indicated by a white cross. At the critical value $\mu_c \approx 2.9$, the selected root switches from the smaller branch $r_n$ to the larger branch $R_n$.

\begin{figure}[!ht]
\center
\includegraphics[scale=1]{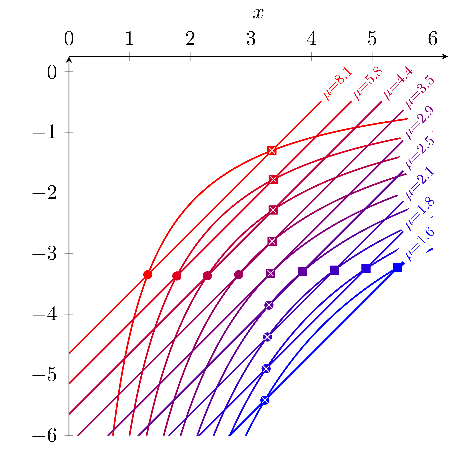}
\caption{Geometric interpretation of the optimality condition for $y=[1;2;3]$. The linear function $L(x)=x-p(\mu)$ and the hyperbolic functions $H_n(x)=-q_{(n)}(\mu)/x$ are shown for increasing values of $\mu$. For each $\mu$, their intersection points define the two real roots $r_n$ (circle) and $R_n$ (square) of the equation 
$x_{(n)}^2-p(\mu)x_{(n)}+q_{(n)}(\mu)=0$. The selected value $x_{(n)}$ is marked by white cross. For $\mu_c=2.9130$ the selection of $x_{(n)}$ switch from $r_n$ to $R_n$.}
\label{cubic_roots_Skuness2}
\end{figure}

\subsection{Determination of the Critical Parameter $\mu_c$}

At the critical point $\mu = \mu_c$, the two potential roots for the largest coordinate coincide, i.e., $r_n = R_n$. Mathematically, this transition is equivalent to the vanishing of the discriminant in the coordinate-wise quadratic equation:
$p^2 = 4q_{(n)}$. In terms of the trigonometric representation, this condition implies $\sin^2 \theta_n = 1$, or equivalently, $\theta_n = \pi/2$. We can establish that at this critical limit, the angles for all coordinates are governed by the relative magnitudes of the input components:
%
\begin{equation}
\sin^2 \theta_i = \frac{y_{(i)}}{y_{(n)}}.
\end{equation}
Consequently, at the transition point $\mu = \mu_c$, the solution components are defined entirely by the small-root branch:
\begin{equation}
x_{(i)} = r_i = p \sin^2 \left( \frac{1}{2} \arcsin \sqrt{\frac{y_{(i)}}{y_{(n)}}} \right).
\end{equation}

By substituting this specific form of $x_{(i)}$ into the definitions of the auxiliary parameters $\alpha$ and $\beta$ \eqref{alpha_beta}, we obtain a coupled system of three equations for the three unknowns $(\alpha, \beta, \mu_c)$. Solving this system provides an explicit, closed-form expression for the critical parameter.
To facilitate a compact representation, we introduce an auxiliary vector $v \in \mathbb{R}^n$ with components defined as:
\begin{equation}
v_i = \sin^2 \left( \frac{1}{2} \arcsin \sqrt{\frac{y_{(i)}}{y_{(n)}}} \right).
\end{equation}
The critical parameter $\mu_c$ is then uniquely determined by the norms of this vector and the maximum magnitude of the input:
\begin{equation}\label{mu_c}
\mu_c = \frac{16 \|v\|_3^6 \left( \|v\|_2^2 - \|v\|_3^3 \right)}{3 \|v\|_2^5} y_{(n)}^2.
\end{equation}

\subsection{Determination of $\alpha$ and $\beta$ for a Given $\mu$}

The computation of the proximity operator requires identifying the specific pair of auxiliary parameters $(\alpha, \beta)$ that satisfy the first-order optimality conditions for a given proximal parameter $\mu$. These parameters are implicitly defined by the norm properties of the solution vector $x$:
\begin{equation}\label{bivar}
\alpha = \frac{\sum x_{(i)}(\alpha,\beta)^2}{\sum x_{(i)}(\alpha,\beta)^3}, \quad \beta^2 = \sum x_{(i)}(\alpha,\beta)^2.
\end{equation}
Substituting the trigonometric representation of the solution components $x_{(i)}$ into these definitions yields a coupled system of nonlinear equations.
To simplify the bivariate system, we introduce a single auxiliary variable $\delta$, which captures the relationship between $\alpha, \beta,$ and $\mu$:
\begin{equation}
\delta = \frac{2\alpha \sqrt{3\mu\beta}}{\beta + 3\mu\alpha}.
\end{equation}
This transformation allows us to express the solution components $x_{(i)}$ as a function of $\delta$ as well as $\alpha$ and $\beta$:
\begin{equation}
x_{(i)}(\alpha,\beta,\delta) 
= \frac{2}{\delta}\sqrt{\frac{\beta}{3\mu\alpha^2}}\sin^2\Big(\frac{1}{2} \arcsin(\delta\sqrt{y_{(i)}})\Big).
\end{equation}
By incorporating $\delta$, we lift the original bivariate problem \eqref{bivar} into a three-dimensional system with three unknowns $(\alpha, \beta, \delta)$. Utilizing the trigonometric representation of the solution components $x_{(i)}$ and defining two auxiliary functions $\phi(\delta)$ and $\psi(\delta)$ that depend only on $\delta$ and the input vector $y$:
\begin{equation}
\phi(\delta) = \frac{\sum \sin^4 \left( \frac{1}{2} \arcsin(\delta \sqrt{y_{(i)}}) \right)}{\sum \sin^6 \left( \frac{1}{2} \arcsin(\delta \sqrt{y_{(i)}}) \right)}, \quad \psi(\delta) = \sum \sin^4 \left( \frac{1}{2} \arcsin(\delta \sqrt{y_{(i)}}) \right),
\end{equation}
the optimality conditions can be expressed as a system of three nonlinear equations:
\begin{equation}\label{ab}
\left\{
\begin{aligned}
\frac{\beta + 3\mu\alpha}{3\mu\alpha^2}
&= \frac{2}{\delta}\sqrt{\frac{\beta}{3\mu\alpha^2}}, \\[6pt]
\frac{2}{\delta}\sqrt{\frac{\beta}{3\mu}}
&= \phi, \\[6pt]
3\mu\alpha^2 \beta \delta^2
&= 4\psi.
\end{aligned}
\right.
\end{equation}

Remarkably, this lifted system admits a closed-form solution for the unknowns in terms of the auxiliary functions $\phi$ and $\psi$. Specifically, the equation for $\delta$ decouples entirely from $\alpha$ and $\beta$, resulting in a single scalar equation for $\delta$:
\begin{equation}\label{delta}
\delta^4 = \frac{16}{3\mu} \cdot \frac{\psi(\delta)^{1/2} (\phi(\delta) - 1)}{\phi(\delta)^3}.
\end{equation}
Solving this equation (e.g., via bisection or Newton's method) yields the optimal parameter $\delta^*$.
Once the optimal $\delta^*$ is identified, the values for the physical norm parameters $\alpha$ and $\beta$ follow immediately through the following closed-form relations derived from the lifted system:
\begin{equation}\label{red_ab}
\alpha = \frac{\sqrt{3 \phi(\delta^*)}\psi(\delta^*)^{1/4}}{3 \sqrt{\mu} \sqrt{\phi(\delta^*) - 1}},
 \quad \beta = \sqrt{3 \mu \phi(\delta^*) (\phi(\delta^*) - 1)} \psi(\delta^*)^{1/4}.
\end{equation}
This reduction effectively transforms the $n$-dimensional nonconvex proximal problem into a stable one-dimensional search, ensuring both computational efficiency and numerical robustness.

Figure \ref{L2L3_toy} illustrates the evolution of $\alpha$ and $\beta$ as functions of $\mu$ for the input vector $y = [1,\,2,\,3]^T$. The vertical surface shows the boundary of infeasible region. The blue markers denote the corresponding values of $\alpha(\mu),\, \beta(\mu)$.

\begin{figure}[!ht]
\center
\includegraphics[scale=0.35,trim={4cm 0 5cm 0},clip]{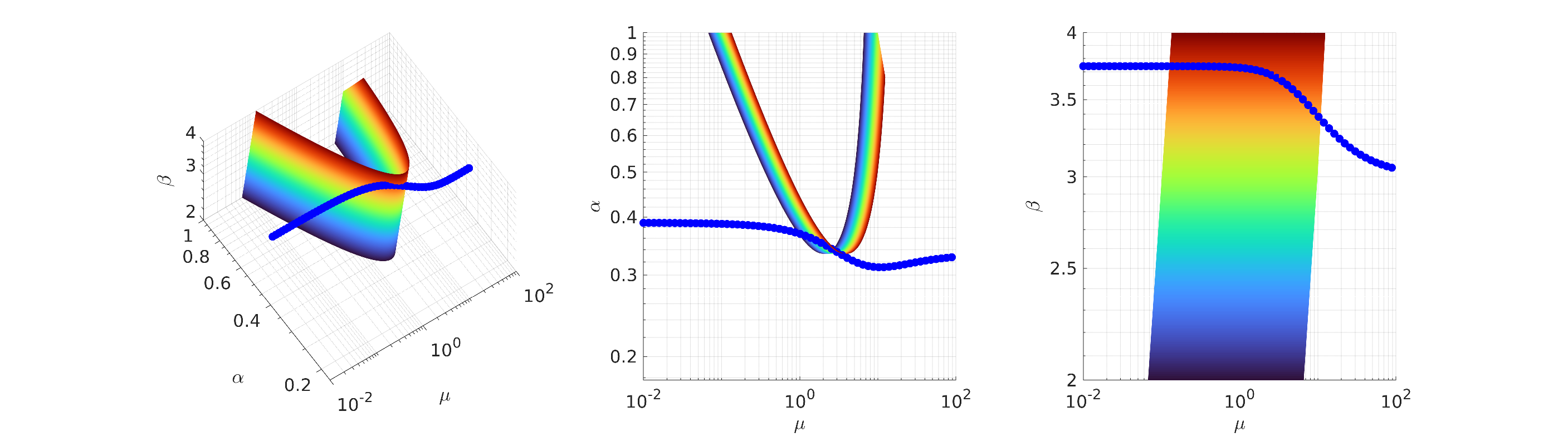}
\caption{The evolution of $\alpha$ and $\beta$ as functions of $\mu$ for the proximity operator of inverse skewness with input vector $y = [1,\,2,\,3]^T$.}
\label{L2L3_toy}
\end{figure}

\subsection{Computational Procedure}
The complete algorithm for computing the proximity operator of the $\ell_2^3/\ell_3^3$ ratio proceeds in $O(n)$ time:
\begin{enumerate}
    \item Sort the absolute values of the input vector $y$ in ascending order.
    \item Calculate the critical parameter $\mu_c$ using \eqref{mu_c}.
    \item Solve the scalar root-finding problem for $\delta^*$ \eqref{delta} to recover the norm parameters $\alpha$ and $\beta$ in \eqref{red_ab}.
    \item Construct $x^*$ by assigning components to the appropriate root branches according to $\mu_c$.
\end{enumerate}

\section{Numerical Examples}

\subsection{Comparison with Global Search}
To empirically validate the calculation of the proximity operators and derivation of the critical parameter $\mu_c$, we conduct a brute-force enumeration test using the input vector $y =[1;~2;~3]$. For this $n=3$ dimensionality, there exist exactly $2^3 = 8$ nonzero candidate vectors formed by every possible combination of the small-root ($r$) and large-root ($R$) branches, in addition to the origin. Tables \ref{tab:validation_skewness} and \ref{tab:validation_kurtosis} represent the result for the inverse skewness and inverse kurtosis proximity operators. For both cases, we evaluate the proximal objective $\Phi(x)$ across representative values of $\mu$ spanning the calculated critical threshold $\mu_c \approx 2.9130$ (skewness) and $\mu_c \approx 0.83$ (kurtosis).

\begin{itemize}
    \item \textbf{Subcritical Regime:} For $\mu < \mu_c$ the global minimizer is unambiguously the \textbf{rrr} pattern, where every component resides on the small-root branch. This confirms that under low regularization, the solution remains a stable deformation of the input signal.
    \item \textbf{Supercritical Regime:} For $\mu > \mu_c$ the global minimum shifts to the \textbf{rrR} pattern. This result perfectly matches the theoretical expectation that as regularization increases, only the largest component $x_{(n)}$ is permitted to switch to the large-root branch to minimize the scale-invariant norm ratio.
\end{itemize}

Crucially, the numerical results demonstrate that all other enumerated combinations (e.g., $Rrr, rRr, RRR$) yield significantly higher objective values. Also, it provides empirical proof for the branch separation property where $R_n > \max_i r_i$ and 
the perfect alignment between global solution obtained via gris search and the proposed analytical framework choosing the root corresponding to the global minimum.

We further validated the proximity operators by plotting the output sample values as a function of the regularization parameter $\mu$. Figure \ref{test_toy_curves} displays these curves for both inverse skewness and inverse kurtosis using the input $y = [1; 2; 3]$. Results from brute-force enumeration (solid lines) and the proposed proximal algorithm (dashed lines) are perfectly superimposed, confirming that our analytical selection rule consistently identifies the global minimizer across all regimes. 
 The plots reveal that the proximal mapping is continuous for all $\mu > 0$. As $\mu$ increases beyond the critical threshold $\mu_c$ (the vertical dashed line), the largest output sample $x_{(3)}$ initially increases as it transitions to the large-root branch before descending to match the input value $y_{(3)}$ in the limit. Conversely, the smaller samples $x_{(1)}$ and $x_{(2)}$ approach zero as $\mu \to \infty$, demonstrating that the solution converges to a maximally sparse (single-spike) state, which is the hallmark of non-Gaussianity maximization.
\begin{table}[!h]
\centering
\caption{Combined Global Minimizer Validation for $\ell_2^3/\ell_3^3$ ($y = [1; 2; 3]$)}
\label{tab:validation_skewness}
\begin{tabular}{cccccc}
\toprule
$\mu$ & \textbf{Type} & \textbf{Pattern} & \textbf{Vector $x^*$} & \textbf{$\Phi(x)$} & \textbf{Global?} \\ \midrule
0.10 & \textbf{Proposed (S)} & \textbf{rrr} & \textbf{[0.98, 1.99, 3.02]} & \textbf{0.15} & \textbf{YES} \\
 & Proposed (L) & rrR & [0.98, 1.99, 82.95] & 3195.94 & No \\
 & Enumerated & rRr & [0.98, 83.98, 3.02] & 3360.28 & No \\
 & Enumerated & rRR & [0.98, 83.98, 82.95] & 6556.16 & No \\
 & Enumerated & Rrr & [84.98, 1.99, 3.02] & 3526.62 & No \\
 & Enumerated & RrR & [84.98, 1.99, 82.95] & 6722.50 & No \\
 & Enumerated & RRr & [84.98, 83.98, 3.02] & 6886.84 & No \\
 & Enumerated & RRR & [84.98, 83.98, 82.95] & 10082.71 & No \\
 & Origin & 0 & [0.00, 0.00, 0.00] & 7.10 & No \\
\midrule
2.91 & \textbf{Proposed (S)} & \textbf{rrr} & \textbf{[0.61, 1.41, 3.32]} & \textbf{3.90} & \textbf{YES} \\
 & Proposed (L) & rrR & [0.61, 1.41, 3.33] & 3.90 & No \\
 & Enumerated & rRr & [0.61, 5.24, 3.32] & 9.29 & No \\
 & Enumerated & rRR & [0.61, 5.24, 3.33] & 9.29 & No \\
 & Enumerated & Rrr & [6.04, 1.41, 3.32] & 16.83 & No \\
 & Enumerated & RrR & [6.04, 1.41, 3.33] & 16.83 & No \\
 & Enumerated & RRr & [6.04, 5.24, 3.32] & 22.72 & No \\
 & Enumerated & RRR & [6.04, 5.24, 3.33] & 22.73 & No \\
 & Origin & 0 & [0.00, 0.00, 0.00] & 9.91 & No \\
\midrule
2.92 & Proposed (S) & rrr & [0.61, 1.40, 3.32] & 3.91 & No \\
 & \textbf{Proposed (L)} & \textbf{rrR} & \textbf{[0.61, 1.40, 3.32]} & \textbf{3.91} & \textbf{YES} \\
 & Enumerated & rRr & [0.61, 5.24, 3.32] & 9.28 & No \\
 & Enumerated & rRR & [0.61, 5.24, 3.32] & 9.28 & No \\
 & Enumerated & Rrr & [6.03, 1.40, 3.32] & 16.80 & No \\
 & Enumerated & RrR & [6.03, 1.40, 3.32] & 16.80 & No \\
 & Enumerated & RRr & [6.03, 5.24, 3.32] & 22.67 & No \\
 & Enumerated & RRR & [6.03, 5.24, 3.32] & 22.67 & No \\
 & Origin & 0 & [0.00, 0.00, 0.00] & 9.92 & No \\
\midrule
5.00 & Proposed (S) & rrr & [0.47, 1.06, 2.05] & 7.61 & No \\
 & \textbf{Proposed (L)} & \textbf{rrR} & \textbf{[0.47, 1.06, 3.37]} & \textbf{6.37} & \textbf{YES} \\
 & Enumerated & rRr & [0.47, 4.37, 2.05] & 9.58 & No \\
 & Enumerated & rRR & [0.47, 4.37, 3.37] & 9.98 & No \\
 & Enumerated & Rrr & [4.96, 1.06, 2.05] & 14.94 & No \\
 & Enumerated & RrR & [4.96, 1.06, 3.37] & 15.34 & No \\
 & Enumerated & RRr & [4.96, 4.37, 2.05] & 18.83 & No \\
 & Enumerated & RRR & [4.96, 4.37, 3.37] & 19.08 & No \\
 & Origin & 0 & [0.00, 0.00, 0.00] & 12.00 & No \\
\bottomrule
\end{tabular}
\end{table}
\begin{table}[!h]
\centering
\caption{Combined Global Minimizer Validation for $\ell_2^4/\ell_4^4$ ($y = [1; 2; 3]$)}
\label{tab:validation_kurtosis}
\begin{tabular}{cccccc}
\toprule
$\mu$ & \textbf{Type} & \textbf{Pattern} & \textbf{Vector $x^*$} & \textbf{$\Phi(x)$} & \textbf{Global?} \\ \midrule
0.10 & \textbf{Proposed (S)} & \textbf{rrr} & \textbf{[0.95, 1.95, 3.05]} & \textbf{0.20} & \textbf{YES} \\
 & Proposed (L) & rrR & [0.95, 1.95, 9.95] & 24.29 & No \\
 & Enumerated & rRr & [0.95, 10.68, 3.05] & 37.80 & No \\
 & Enumerated & rRR & [0.95, 10.68, 9.95] & 62.06 & No \\
 & Enumerated & Rrr & [11.27, 1.95, 3.05] & 52.87 & No \\
 & Enumerated & RrR & [11.27, 1.95, 9.95] & 77.13 & No \\
 & Enumerated & RRr & [11.27, 10.68, 3.05] & 90.64 & No \\
 & Enumerated & RRR & [11.27, 10.68, 9.95] & 114.90 & No \\
 & Origin & 0 & [0.00, 0.00, 0.00] & 7.10 & No \\
\midrule
0.82 & \textbf{Proposed (S)} & \textbf{rrr} & \textbf{[0.74, 1.58, 3.26]} & \textbf{1.44} & \textbf{YES} \\
 & Proposed (L) & rrR & [0.74, 1.58, 3.29] & 1.44 & No \\
 & Enumerated & rRr & [0.74, 4.72, 3.26] & 5.26 & No \\
 & Enumerated & rRR & [0.74, 4.72, 3.29] & 5.27 & No \\
 & Enumerated & Rrr & [5.26, 1.58, 3.26] & 10.76 & No \\
 & Enumerated & RrR & [5.26, 1.58, 3.29] & 10.77 & No \\
 & Enumerated & RRr & [5.26, 4.72, 3.26] & 15.00 & No \\
 & Enumerated & RRR & [5.26, 4.72, 3.29] & 15.02 & No \\
 & Origin & 0 & [0.00, 0.00, 0.00] & 7.82 & No \\
\midrule
0.84 & Proposed (S) & rrr & [0.74, 1.57, 3.25] & 1.47 & No \\
 & \textbf{Proposed (L)} & \textbf{rrR} & \textbf{[0.74, 1.57, 3.27]} & \textbf{1.47} & \textbf{YES} \\
 & Enumerated & rRr & [0.74, 4.69, 3.25] & 5.22 & No \\
 & Enumerated & rRR & [0.74, 4.69, 3.27] & 5.24 & No \\
 & Enumerated & Rrr & [5.24, 1.57, 3.25] & 10.67 & No \\
 & Enumerated & RrR & [5.24, 1.57, 3.27] & 10.68 & No \\
 & Enumerated & RRr & [5.24, 4.69, 3.25] & 14.86 & No \\
 & Enumerated & RRR & [5.24, 4.69, 3.27] & 14.87 & No \\
 & Origin & 0 & [0.00, 0.00, 0.00] & 7.84 & No \\
\midrule
2.50 & Proposed (S) & rrr & [0.51, 1.07, 1.80] & 5.83 & No \\
 & \textbf{Proposed (L)} & \textbf{rrR} & \textbf{[0.51, 1.07, 3.37]} & \textbf{3.74} & \textbf{YES} \\
 & Enumerated & rRr & [0.51, 3.91, 1.80] & 6.28 & No \\
 & Enumerated & rRR & [0.51, 3.91, 3.37] & 7.00 & No \\
 & Enumerated & Rrr & [4.27, 1.07, 1.80] & 10.20 & No \\
 & Enumerated & RrR & [4.27, 1.07, 3.37] & 10.94 & No \\
 & Enumerated & RRr & [4.27, 3.91, 1.80] & 13.74 & No \\
 & Enumerated & RRR & [4.27, 3.91, 3.37] & 14.47 & No \\
 & Origin & 0 & [0.00, 0.00, 0.00] & 9.50 & No \\
\bottomrule
\end{tabular}
\end{table}

\begin{figure}[!ht]
\center
\includegraphics[scale=0.55]{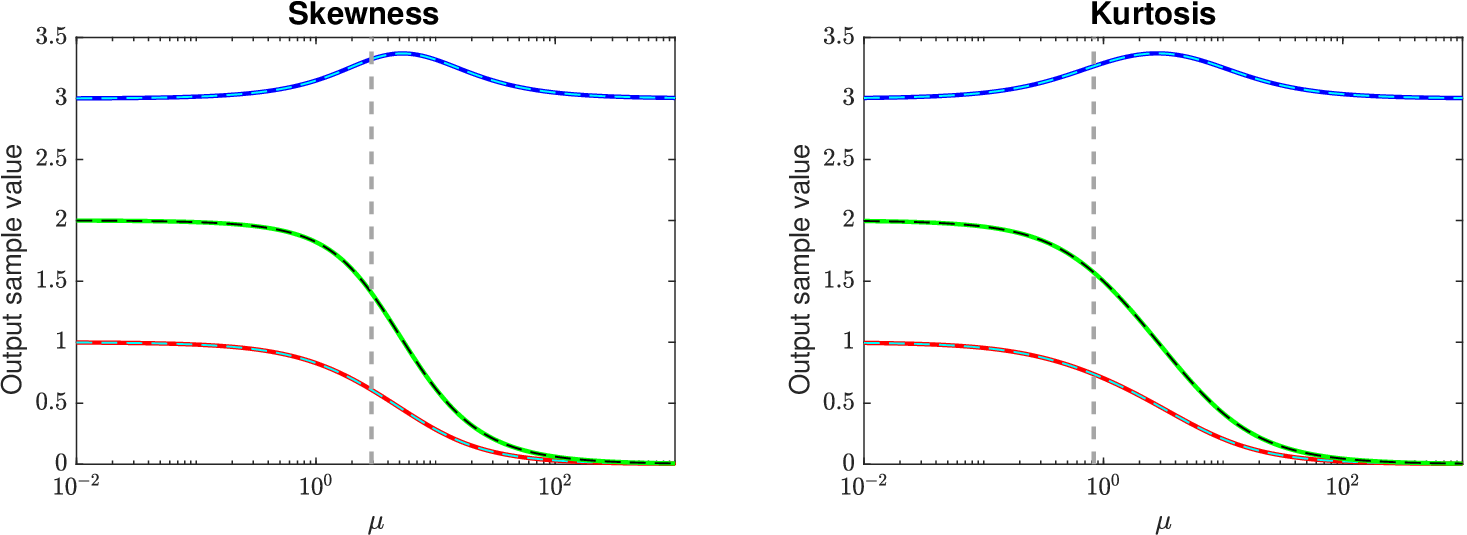}
\caption{Numerical validation and sample variation of the proximity operators for the input vector $y = [1; 2; 3]$. Left: Results for inverse skewness ($\ell_2^3/\ell_3^3$). Right: Results for inverse kurtosis ($\ell_2^4/\ell_4^4$). Solid curves represent the global minimum identified through brute-force enumeration of all $2^3=8$ non-zero stationary points, while dashed curves denote the solution obtained via the proposed analytical proximal algorithm. The perfect superposition of the curves validates the derived selection rule and the critical threshold parameter $\mu_c$ (vertical dashed line).}
\label{test_toy_curves}
\end{figure}


\subsection{Wavelet Phase Correction}
In this section, we evaluate the performance of our method for the phase correction of a Ricker wavelet. A symmetric (zero-phase) Ricker wavelet with a dominant frequency of 3 Hz was sampled with $dt=1$ ms. Constant phase shifts of $\pm 60^\circ$ were applied, as shown in the top row of Figure \ref{Ricker_cphase}. We then applied the proposed ADMM-based phase correction \eqref{ADMM_phase} to estimate the phase, employing first-order Tikhonov regularization to stabilize the estimates. The estimated phases for both inverse skewness and inverse kurtosis, shown in the second row of the figure, are highly accurate and correspond to the negative of the phase shifts added to the original signal. The final rotated wavelets (shown in red in the third row) are consistent with the original zero-phase wavelets (shown in blue). Finally, the skewness and kurtosis curves in the bottom row confirm a consistent increase in these statistical values at each iteration, demonstrating the stable convergence of the optimization. 

To further evaluate the framework's performance in a nonstationary environment, we cascaded two Ricker wavelets with distinct initial phase shifts, as illustrated in Figure \ref{Ricker_twins}. This test is designed to verify the algorithm's ability to track phase variations that change over time between different seismic events. Again, both the skewness and kurtosis methods successfully estimated the local phase functions and corrected the distortions.

\begin{figure}[!ht]
\center
\includegraphics[scale=0.45]{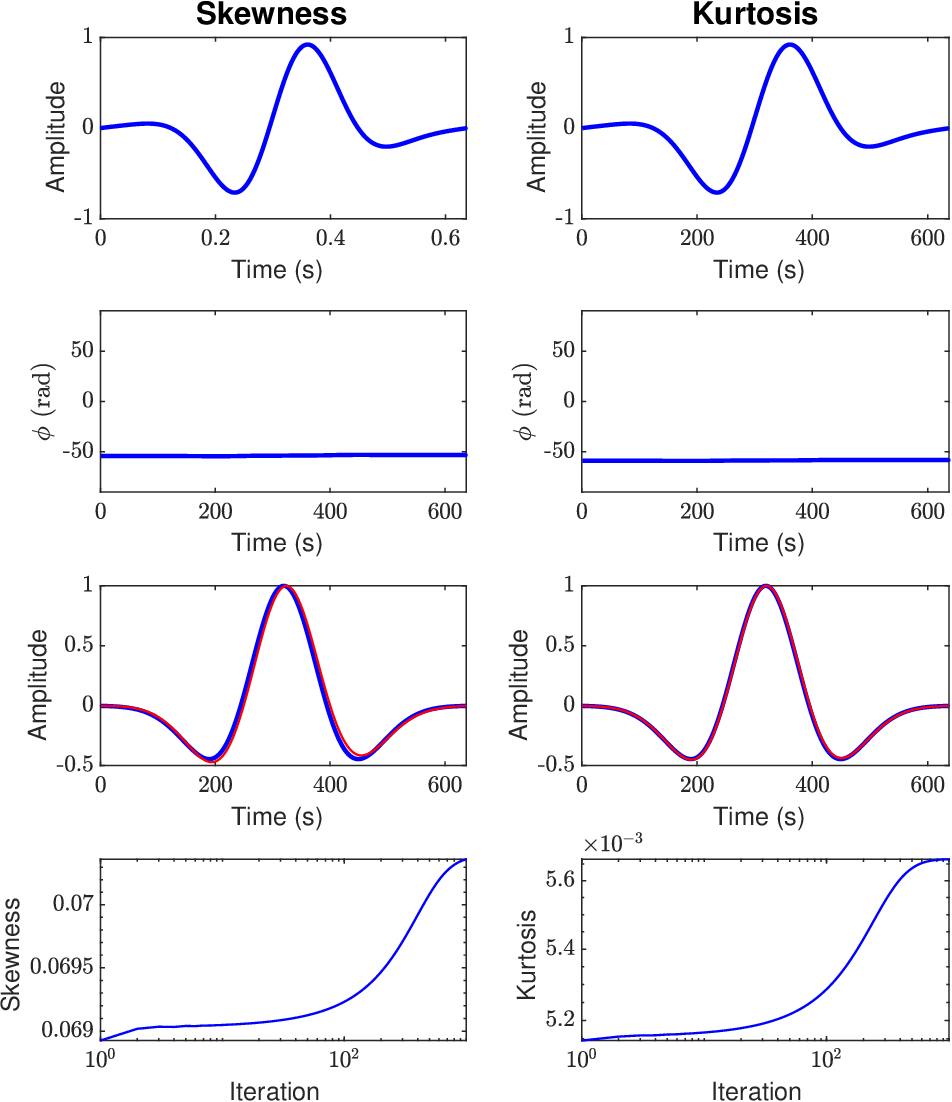}
\includegraphics[scale=0.45]{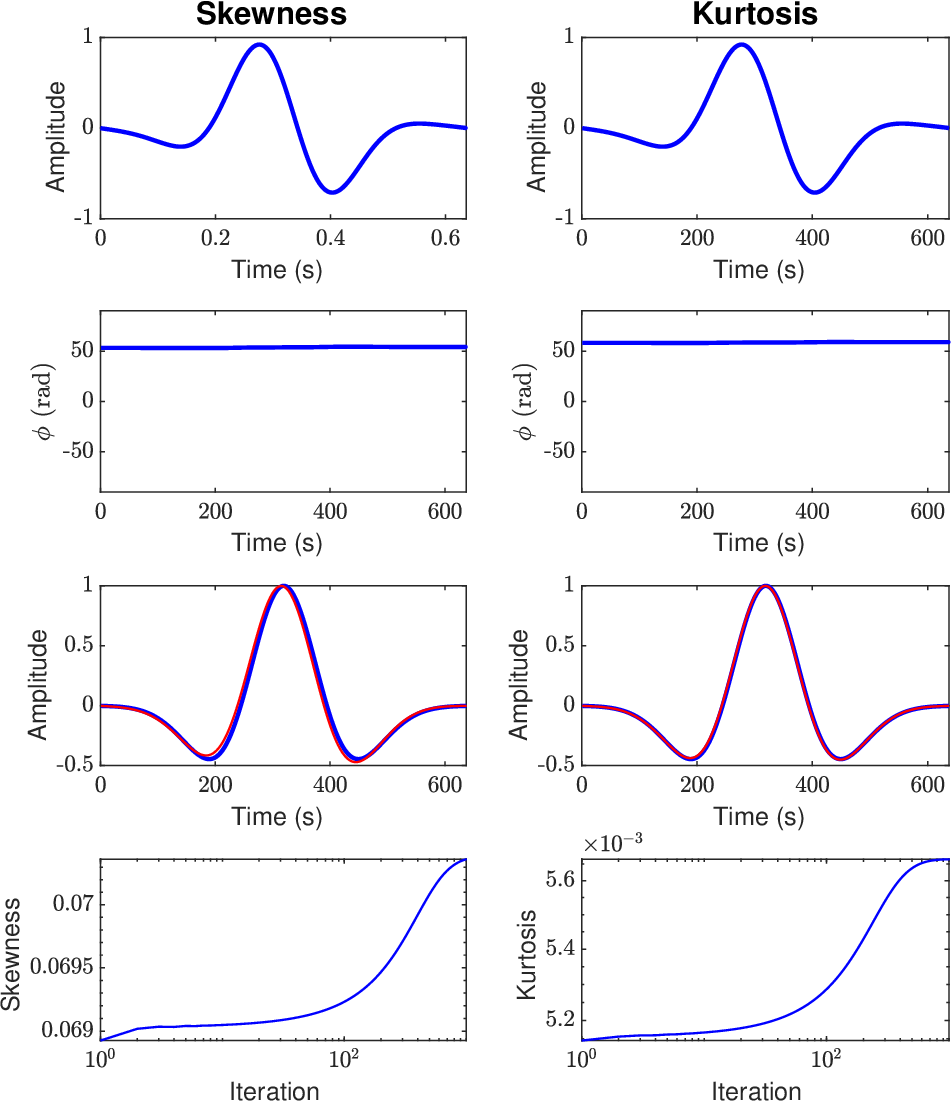}

\caption{Nonstationary phase correction of a 3 Hz Ricker wavelet using the proposed ADMM-based framework. (Top row) Input wavelets corrupted with constant phase shifts of $+60^\circ$ (left) and $-60^\circ$ (right). (Second row) Estimated phase functions obtained via skewness and kurtosis. (Third row) Comparison of the final rotated wavelets (red) against the original zero-phase wavelets (blue). (Bottom row) Convergence curves of the skewness and kurtosis values across ADMM iterations.}
\label{Ricker_cphase}
\end{figure}

\begin{figure}[!ht]
\center
\includegraphics[scale=0.45]{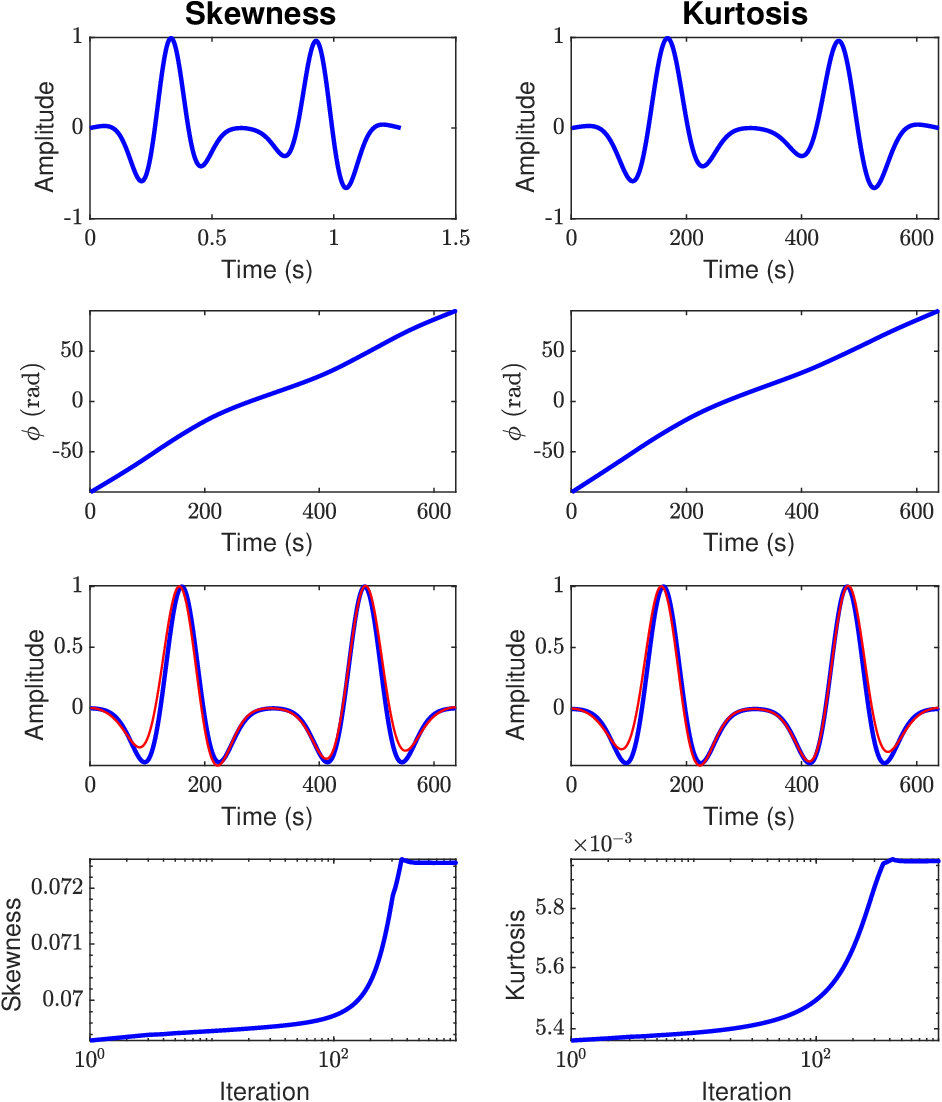}
\caption{ Nonstationary phase correction of Ricker wavelets with different phase shifts. (Top row) Input trace containing two wavelets. (Second row) Estimated nonstationary phase using skewness and kurtosis methods. (Third row) Final corrected trace (red) compared to the original zero-phase model (blue). (Bottom row)  Convergence curves of the skewness and kurtosis values across ADMM iterations.}
\label{Ricker_twins}
\end{figure}

\section{Conclusions}

In this paper we introduced a robust ADMM-based framework for nonstationary seismic phase estimation by recasting higher-order statistical maximization as a regularized proximal optimization task. By deriving the first closed-form proximity operators for scale-invariant inverse kurtosis and inverse skewness, we reduce complex high-dimensional non-convex minimizations to efficient $O(n)$ 1D root-finding subproblems. Numerical experiments using synthetic seismic data confirm that the proposed method accurately recovers nonstationary phase fields and can provide phase corrected signals. This method may be used as an alternative tool to obtain improved results in post-stack seismic data interpretation.

%

\section*{Acknowledgments}
We would like to acknowledge the assistance of volunteers in putting
together this example manuscript and supplement.

\bibliographystyle{siamplain}
\newcommand{\SortNoop}[1]{}

\end{document}

%% file: ex_shared.tex
\usepackage{lipsum}
\usepackage{amsfonts}
\usepackage{graphicx}
\usepackage{epstopdf}
\usepackage{algorithmic}
\ifpdf
  \DeclareGraphicsExtensions{.eps,.pdf,.png,.jpg}
\else
  \DeclareGraphicsExtensions{.eps}
\fi

\newsiamremark{remark}{Remark}
\newsiamremark{hypothesis}{Hypothesis}
\crefname{hypothesis}{Hypothesis}{Hypotheses}
\newsiamthm{claim}{Claim}
\newsiamremark{fact}{Fact}
\crefname{fact}{Fact}{Facts}

\headers{Proximal Maximization of Skewness and Kurtosis}{A. Gholami}

\title{Regularized Nonstationary Phase Estimation via Proximal Maximization of Skewness and Kurtosis
}

\author{Ali Gholami\thanks{Institute of Geophysics, Polish Academy of Sciences, Warsaw, Poland 
  (\email{agholami@igf.edu.pl}).}
}

\usepackage{amsopn}
